\documentclass[11pt]{article}

\usepackage{amsmath,amsthm,amssymb,amsfonts,verbatim}
\usepackage{fullpage}

\usepackage[dvipsnames,usenames]{color}
\usepackage[pagebackref]{hyperref}
\usepackage{times}

 \providecommand{\F}{\mathbb{F}}

\title{Optimal locally repairable codes of distance $3$ and $4$ via cyclic codes}

\date{}

\newtheorem{lemma}{Lemma}[section]
\newtheorem{theorem}[lemma]{Theorem}

\newtheorem{ex}[lemma]{Example}
\newtheorem{defn}{Definition}

\theoremstyle{remark}
\newtheorem{rmk}{Remark}
\newtheorem{open}{Open problem}

\renewcommand{\epsilon}{\varepsilon}
\renewcommand{\le}{\leqslant}
\renewcommand{\ge}{\geqslant}

\newcommand{\vnote}[1]{}

\def \Xi {{X^{[i]}}}

\newcommand{\Ga}{\alpha}
\newcommand{\Gb}{\beta}
\newcommand{\Gg}{\gamma}     
\newcommand{\Gd}{\delta}

\newcommand{\Gs}{\sigma}

\def \ba {{\bf a}}

\def \bc {{\bf c}}

\def\Aut {{\rm Aut }}

\def\LRC {{\rm locally repairable code\ }}
\def\LRCs {{\rm locally repairable codes\ }}

\begin{document}

\author{Yuan Luo\thanks{Department of Computer Sciences and Engineering, Shanghai Jiaotong University, Shanghai 200240, P. R. China. { (email: luoyuan@cs.sjtu.edu.cn)}}, Chaoping Xing and Chen Yuan
\thanks{The authors are with Division of Mathematical Sciences, School
of Physical and Mathematical Sciences, Nanyang Technological
University, Singapore 637371, Republic of Singapore
(email: \{xingcp,YUAN0064\}@ntu.edu.sg).}
 }

\maketitle

\begin{abstract}
Like classical block codes, a locally repairable code also obeys the Singleton-type bound (we call a locally repairable code {\it optimal} if it achieves the Singleton-type bound). In the breakthrough work of \cite{TB14}, several classes of optimal locally repairable codes were constructed via subcodes of Reed-Solomon codes. Thus, the  lengths of the codes given in \cite{TB14} are upper bounded by the code alphabet size $q$. Recently, it was proved through extension of construction in \cite{TB14}   that  length of $q$-ary optimal locally repairable codes can be $q+1$ in \cite{JMX17}. Surprisingly, \cite{BHHMV16} presented a few examples of $q$-ary optimal locally repairable codes of small distance and locality with code length achieving roughly $q^2$. Very recently, it was further shown in \cite{LMX17} that there exist $q$-ary optimal locally repairable codes with length bigger than $q+1$ and  distance propositional to $n$.
 Thus, it becomes an interesting and challenging problem to construct new families of $q$-ary optimal locally repairable codes of length bigger than $q+1$.

 In this paper, we construct a class of  optimal locally repairable codes of distance $3$ and $4$ with unbounded length (i.e., length of the codes is independent of the code alphabet size). Our technique is through cyclic codes with particular generator and parity-check polynomials that are carefully chosen.
\end{abstract}
\section{Introduction}
Due to applications to distributed storage systems,  locally repairable codes  have recently attracted great attention of researchers \cite{HL07,GHSY12,PKLK12,SRKV13,JMX17,LMX17,FY14,PD14,TB14,TPD16,BTV17}. A local repairable code is nothing but a  block code with an additional parameter called {\it locality}.  For a locally repairable code $C$ of length $n$ with $k$ information symbols and locality $r$ (see the definition of locally repairable codes in Section \ref{subsec:2.1}), it was proved in \cite{GHSY12} that the minimum distance $d(C)$ of $C$ is upper bounded by
 \begin{equation}\label{eq:x1}
 d(C)\le n-k-\left\lceil \frac kr\right\rceil+2.
 \end{equation}
 The bound \eqref{eq:x1} is called the Singleton-type bound for locally repairable codes and was proved by extending the arguments in the proof of the classical Singleton bound on codes.  In this paper, we refer an optimal locally repairable code to a block code achieving the bound \eqref{eq:x1}.
 \subsection{Known results}
 The early constructions of optimal \LRCs gave codes with alphabet size that is exponential in code length (see \cite{HCL,SRKV13}. There  was also an earlier construction of optimal locally repairable codes given in \cite{PKLK12} with  alphabet  size comparable to code length. However, the construction in \cite{PKLK12} only produced  a specific value of the length $n$, i.e., $n=\left\lceil \frac kr\right\rceil(r+1)$. Thus, the rate of the code is very close to $1$. There are also some existence results given in \cite{PKLK12} and \cite{TB14} with less restriction on locality $r$. But the results in both the papers require large alphabet size which is an exponential function of the code length.

 A recent breakthrough construction  given  in \cite{TB14} makes use of subcodes of Reed-Solomon codes. This construction produces  optimal \LRCs with length linear  in  alphabet size although the length of codes is upper bounded by alphabet size.  The construction in \cite{TB14} was extended via automorphism group of rational function fields by Jin, Ma and Xing \cite{JMX17} and it turns out that there are more flexibility on locality and the code length can go up to $q+1$, where $q$ is the alphabet size.

 Based on the classical MDS conjecture, one should wonder if $q$-ary optimal  locally repairable codes can have length bigger than $q+1$. Surprisingly, it was shown in \cite{BHHMV16} that there exist $q$-ary optimal locally repairable codes of length exceeding  $q+1$.  Although \cite{BHHMV16} produced a few optimal  locally repairable codes with specific parameters,
 it paves a road for people to continue search for such optimal codes. Very recently, it was shown in \cite{LMX17} via elliptic curves that there exist $q$-ary optimal  locally repairable codes with length $n$ bigger than $q+1$ and distance proportional to  length $n$.

 \subsection{Our result}
 In this paper, by carefully choosing generator polynomials, we can show that there exist   optimal  locally repairable codes of distance $3$ and $4$ that are cyclic. One feature of these codes is that their lengths are unbounded, i.e., lengths are independent of the code alphabet sizes.  More precisely, we have the following main result in this paper.
\begin{theorem}\label{thm:1.1} Let $q$ be a prime power. Assume that $n$ is a  positive integer with $\gcd(n,q)=1$. Then
\begin{itemize}
\item[{\rm (i)}]  there is a $q$-ary   optimal $\left[n,k=n-\left(1+\frac n{r+1}\right),3\right]$ cyclic locally repairable  code with locality $r$
if $r\ge 2$ and $\gcd(n,q-1)\equiv 0\pmod{r+1}$; and
\item[{\rm (ii)}]  there is a $q$-ary   optimal $\left[n,k=n-\left(2+\frac n{r+1}\right),4\right]$ cyclic locally repairable  code with locality $r$
if $r\ge 3$, $\gcd(n,q-1)\equiv 0\pmod{r+1}$ and $\gcd\left(\frac{n}{r+1},r+1\right)$ divides $2$.
\end{itemize}
\end{theorem}
 \begin{rmk}
 \begin{itemize}
\item[{\rm (i)}]  The conditions in Theorem \ref{thm:1.1} are easy to be satisfied and the length $n$ is unbounded.
\item[{\rm (ii)}]  Indeed, parameters given in Theorem \ref{thm:1.1} satisfy the equality of \eqref{eq:x1}. Let us verify this only for part (i) of Theorem \ref{thm:1.1}. We have
\[n-k-\left\lceil \frac kr\right\rceil+2=1+\frac{n}{r+1}-\left\lceil\frac{ n-1-\frac{n}{r+1}}r\right\rceil+2=3+\frac{n}{r+1}-\left\lceil\frac{n}{r+1}-\frac{1}r\right\rceil=3.\]
\end{itemize}
 \end{rmk}

 \subsection{Open problems}
In view of the known results and our result in this paper, all known optimal $q$-ary \LRCs with length $n$ and minimum distance  satisfy either (i) $d$ is small and $n$ is unbounded; or (ii) $d$ is proportional to $n$ and $q$ is linear in $n$. One natural  question is
\begin{open} Are there optimal $q$-ary \LRCs with length $n$ much bigger than $q$ (for instance $n=\Omega(q^2)$) and minimum distance proportional to $n$?
\end{open}
The other open problem is the following.
\begin{open} Are there  optimal $q$-ary \LRCs with unbounded length $n$ for every constant $d$.
\end{open}
 \subsection{Organization of the paper}
 The paper is organized as follows. In Section 2, we provide some preliminaries on \LRCs and cyclic codes. In Section 3, we  present proof of  Theorem \ref{thm:1.1}. In addition, we also give a construction of optimal $q$-ary \LRCs of length $2(q-1)$ and distance $4$.
\section{Preliminaries}
In this section, we present some preliminaries on locally repairable codes and cyclic codes.
\subsection{Locally repairable codes}\label{subsec:2.1}
Informally speaking, a block code is said with locality $r$ if  every coordinate of a given codeword can be recovered by accessing at most $r$ other coordinates of this codeword. Precisely speaking,  a locally repairable code with locality $r$ is given as follows.
\begin{defn}
Let $C\subseteq \F_q^n$ be a $q$-ary block code of length $n$. For each $\Ga\in\F_q$ and $i\in \{1,2,\cdots, n\}$, define $C(i,\Ga):=\{\bc=(c_1,\dots,c_n)\in C\; | \; c_i=\Ga\}$. For a subset $I\subseteq \{1,2,\cdots, n\}\setminus \{i\}$, we denote by $C_{I}(i,\Ga)$ the projection of $C(i,\Ga)$ on $I$.
Then $C$ is called a locally repairable code with locality $r$ if, for every $i\in \{1,2,\cdots, n\}$, there exists a subset
$I_i\subseteq \{1,2,\cdots, n\}\setminus \{i\}$ with $|I_i|\le r$ such that  $C_{I_i}(i,\Ga)$ and $C_{I_i}(i,\Gb)$ are disjoint for any $\Ga\neq \Gb\in\F_q$.
\end{defn}
Apart from the usual parameters: length, rate and minimum distance,  the locality of a   locally repairable code plays a crucial role. In this paper, we always consider locally repairable codes that are linear over $\F_q$. Thus, a $q$-ary \LRC of length $n$, dimension $k$, minimum distance $d$ and locality $r$ is said to be an $[n,k,d]_q$-\LRC with locality $r$.
\subsection{Cyclic codes}
Cyclic codes are well known and understood in the coding community. We briefly introduce them and list some useful facts on cyclic codes in this subsection. A $q$-ary cyclic code $C$ is identified with an ideal of the ring $\F_q[x]/(x^n-1)$. As every ideal of $\F_q[x]/(x^n-1)$ is a principal ideal generated by a divisor $g(x)$ of $x^n-1$, we can just write $C=\langle g(x)\rangle$. The dimension of $C$ is $n-\deg(g(x))$. Furthermore, $c(x)$ is a codeword of $C$ if and only if $c(x)$ is divisible by $g(x)$. In other words, if $\Ga_1,\dots,\Ga_{n-k}\in\bar{\F}_q$ are all $n-k$ roots of $g(x)$ (where $\bar{\F}_q$ stands for algebraic closure of $\F_q$), then $c(x)$ belongs to $C$ if and only if $c(\Ga_i)=0$ for all $i=1,2,\dots,n-k$.

Let $C=\langle g(x)\rangle$ for a divisor $g(x)$ of $x^n-1$. Put $h(x)=\frac{x^n-1}{g(x)}$. Let $\tilde{h}(x)$ be the reciprocal polynomial of $h(x)$. Then the Euclidean dual code  of $C$ is $C^{\perp}=\langle\tilde{h}(x)\rangle$. We recall some facts about cyclic codes without giving proof. The reader may refer to the books \cite{MS,LX}.
\begin{lemma}\label{lem:2.1}
\begin{itemize}
\item[{\rm (i)}] Let $h(x)$ be  a divisor of $x^n-1$, then the codes $\langle{h}(x)\rangle$ and $\langle\tilde{h}(x)\rangle$ are equivalent.
\item[{\rm (ii)}] Let $C=\langle g(x)\rangle$ a $q$-ary cyclic code of length $n$ for a divisor $g(x)$ of $x^n-1$ and assume that $\Ga_1,\dots,\Ga_t$ are roots of $g(x)$. If $C$ has a nonzero codeword $c(x)=\sum_{j=1}^tc_{i_j}x^{i_j}$ of weight at most $t$, then the determinant
\[\det\left(\Ga_\ell^{i_j}\right)_{1\le \ell\le t,1\le j\le t}\]
is zero.
\item[{\rm (iii)}] Let $C=\langle g(x)\rangle$ be a $q$-ary cyclic code of length $n$ for a divisor $g(x)$ of $x^n-1$. Let $\beta\in\bar{\F}_q$ be an $n$th primitive  root of unity. If there exist an integer $t$ and a positive integer $\Gd$ such that $\Gb^t,\Gb^{t+1},\dots,\Gb^{t+\Gd-2}$ are roots of $g(x)$, then $C$ has minimum distance at least $\Gd$.
\end{itemize}
\end{lemma}
Let $S_n$ denote the symmetric group of length $n$. An automorphism of a $q$-ary block code $C$ is a permutation $\Gs\in S_n$ satisfying that $(c_{\Gs(1)},c_{\Gs(2)},\dots,c_{\Gs(n)})\in C$ whenever $(c_1,c_2,\dots,c_n)\in C$. All automorphisms $C$ form a subgroup of $S_n$, denoted by $\Aut(C)$. It is called the {\it automorphism group} of $C$. The code $C$ is called {\it transitive} if, for any $i,j\in\{1,2,\dots,n\}$, there exists an automorphism $\Gs\in \Aut(C)$ such that $\Gs(i)=j$.

If $C$ is the cyclic code, then $\Aut(C)$ contains the subgroup generated by the cyclic shift $(12\cdots n)$. Thus, $C$ is transitive.

\section{Constructions}
In this section, we first give a general result on locality for transitive codes and then apply it to the proof of Theorem \ref{thm:1.1}. In addition, we also present a construction of $q$-ary $[2(q-1),4]$ \LRC with locality $r$.
\subsection{A general result}
\begin{lemma}\label{lem:3.1} If $C$ is a $q$-ary transitive linear code with dual distance $d^{\perp}\le r+1$, then $C$ has locality $r$.
\end{lemma}
\begin{proof} Let $\ba=(a_1,a_2,\dots,a_n)$ be a codeword of $C^{\perp}$ with the Hamming weight $d^{\perp}$. Let $I$ denote the support of $\ba$. Then $|I|= d^{\perp}\le r+1$. Let $J$ be a subset of $\{1,2,\dots,n\}$ such that $|J|=r+1$ and $I\subseteq J$. Choose an element $i\in I$. Then for every codeword $\bc=(c_1,c_2,\dots,c_n)\in C$, we have $0=\sum_{j\in I}a_jc_j=\sum_{j\in J}a_jc_j$. This gives $c_i=-a_i^{-1}\sum_{j\in J\setminus\{i\}}a_jc_j$. This implies that $c_i$ can be repaired by $\{c_j\}_{j\in J\setminus\{i\}}$. Now for any $\ell\in \{1,2,\dots,n\}$, there exists an automorphism $\Gs\in S_n$ such that $\Gs(i)=\ell$. This implies that
\[0=\sum_{j\in I}a_jc_{\Gs(j)}=\sum_{j\in J}a_jc_{\Gs(j)}.\]
Hence, $c_\ell=c_{\Gs(i)}=-a_i^{-1}\sum_{j\in J\setminus\{i\}}a_jc_{\Gs(j)}$. The proof is completed.
\end{proof}
\begin{ex}{\rm Let $q$ be a prime power. Let $r,n$ be positive integers with $r\ge 2$, $n|(q-1)$ and $(r+1)|n$. We present a $q$-ary optimal cyclic locally repairable code of length $n$ minimum distance $d$ and locality $r$ for any $1\le d\le n$.
Although such a code was already given in~\cite{TB14}, we provide different view for this code.
Such an observation finally leads to  discovery of  optimal locally repairable codes with unbounded length given in Theorem \ref{thm:1.1}.
Let $\beta$ be an $n$th primitive  root of unity. For a positive integer $d$ with $1\le d\le n$, $d$ can be uniquely written as  $d=(r+1)a+b$ for some integer $a\ge 0$ and integer  $b\in \{0,1,\ldots,r\}$.

{\bf Case 1.} $b\in \{2,3,4,\ldots,r\}$.

Let
$$g(x)=\prod_{i=0}^{d-2}(x-\beta^i)\times\prod_{j=a+1}^{\frac{n}{r+1}-1}(x-\beta^{(r+1)j+b-2}).$$
It is obvious that $g(x)|(x^n-1)$ since all roots of $g(x)$ are the $n$th roots of unity and they are distinct.
We denote by $C$ the cyclic code with generator polynomial $g(x)$. Due to the fact that $g(x)$ contains the
$d-1$  roots $1,\beta,\ldots,\beta^{d-2}$, it follows from Lemma \ref{lem:2.1}(i) that $C$ has the minimum distance at least $d$. Moreover, the dimension of $C$ is
$k:=n-\deg(g(x))=n-\frac{n}{r+1}-(d-a-2)=\frac{rn}{r+1}-ar-b+2$.
Thus, we have
\[n-k-\left\lceil \frac kr\right\rceil+2=\frac{n}{r+1}+ar+b-2-\left(\frac{n}{r+1}-a+\left\lceil\frac{-b+2}{r}\right\rceil\right)+2=a(r+1)+b=d.\]
Thus,
to show that $C$ satisfies the Singleton-type bound \eqref{eq:x1}, it remains to show that the locality of $C$ is $r$. By Lemma \ref{lem:3.1}, it is sufficient to the dual distance of $C$ is at most $r+1$. By Lemma \ref{lem:2.1}(i), it is sufficient to show that the cyclic code $\langle h(x)\rangle$ has minimum distance at most $r+1$, where  $h(x)=\frac{x^n-1}{g(x)}$.
Observe that
$$\prod_{j=0}^{\frac{n}{r+1}-1}(x-\beta^{(r+1)j+b-2})=x^{\frac{n}{r+1}}-\beta^{\frac{(b-2)n}{r+1}}.$$
Thus, we write
$$
g(x)=(x^{\frac{n}{r+1}}-\beta^{\frac{(b-2)n}{r+1}})\prod_{i=0,i\neq b-2\bmod{r+1}}^{d-2}(x-\beta^i).
$$
Then \[h(x)\prod_{i=0,i\neq b-2\bmod{r+1}}^{d-2}(x-\beta^i)=\frac{x^n-1}{x^{\frac{n}{r+1}}-\beta^{\frac{(b-2)n}{r+1}}}=
\sum_{i=0}^r\beta^{\frac{(b-2)in}{r+1}}x^{\frac{(r-i)n}{r+1}}\]
is a codeword of $\langle h(x)\rangle$ and it has Hamming weight $r+1$. This implies that $C^{\perp}$ has minimum distance at most $r+1$.

{\bf Case 2.} $b\in \{0,1\}$.

In this case, we define the cyclic code $C$ generated by
$$g(x)=\prod_{i=0}^{d-2}(x-\beta^i)\times\prod_{j=a+1}^{\frac{n}{r+1}}(x-\beta^{(r+1)j+b-2}).$$
Compared with Case $1$, the degree of $g(x)$ increases by $1$ due to
fact that $b-2$ is negative.
We can mimic the proof of Case 1 and skip the detail.
}\end{ex}

\begin{ex}
{\rm Let $q$ be a prime power. Let $r,n$ be positive integers with $r\ge 2$, $n|(q+1)$ and $(r+1)|n$. We present a $q$-ary optimal cyclic locally repairable code of length $n$, locality $r$ and minimum distance $d=a(r+1)+b$ with $2|a$ and $b\in\{2,4,6,\ldots,2\lceil \frac{r-1}{2}\rceil\}$. The codes in this example were founded in \cite{JMX17} already, but we provide a new view on the construction.

Since $n|(q+1)$, we have $(x^n-1)|(x^{q^2-1}-1)$. This implies that each irreducible factor of $x^n-1$ is either linear or quadratic. If an irreducible factor $f(x)$ of $x^n-1$ is  quadratic and $\Gg$ is a root of $f(x)$, then the other root is $\Gg^q$.  As $n|(q+1)$ and $\Gg^n=1$, we have $\Gg^{q+1}=1$. This gives $\Gg^q=\Gg^{-1}$. In conclusion, for each $n$th root of unity $\gamma$, $(x-\gamma)(x-\gamma^{-1})$ is a polynomial over $\F_q$.

Let $\beta$ be an $n$th primitive  root of unity.
Let
$$g(x)=\prod_{i=-\frac{d-2}{2}}^{\frac{d-2}{2}}(x-\beta^i)\times\prod_{j=\frac{a+2}{2}}^{\frac{n}{r+1}- \frac{a+2}{2}}(x-\beta^{(r+1)j}).$$
We first show that $g(x)$ is defined over $\F_q$. If $\beta^i$ is a root of $g(x)$, then either
$-\frac{d-2}{2}\le i\le \frac{d-2}{2}$ or $(r+1)|i$. This implies that both $\beta^i$ and $\beta^{-i}$ are the roots of $g(x)$.
Thus, $g(x)$ is the product of $x-1$ with some polynomials of the form $(x-\beta^i)(x-\beta^{-i})$. We conclude that $g(x)$ is defined over $\F_q$. It is clear that $g(x)|(x^n-1)$ since all roots of $g(x)$ are  $n$th roots of unity and they are distinct.
Denote by $C$ the cyclic code with generator polynomial $g(x)$.
Due to the fact that $g(x)$ contains the
$d-1$  roots $\beta^{-\frac{d-2}{2}},\ldots,\beta^{\frac{d-2}{2}}$, it follows from Lemma \ref{lem:2.1}(i) that $C$ has the minimum distance at least $d$.
Moreover, the dimension of $C$ is
$k:=n-\deg(g(x))=n-\frac{n}{r+1}-(d-a-2)=\frac{rn}{r+1}-ar-b+2$.
Thus, we have
\[n-k-\left\lceil \frac kr\right\rceil+2=\frac{n}{r+1}+ar+b-2-\left(\frac{n}{r+1}-a+\left\lceil\frac{-b+2}{r}\right\rceil\right)+2=a(r+1)+b=d.\]
To show that $C$ satisfies the Singleton-type bound \eqref{eq:x1}, it remains to show that the locality of $C$ is $r$. By Lemma \ref{lem:3.1}, it is sufficient to show that the dual distance of $C$ is at most $r+1$. By Lemma \ref{lem:2.1}(i), it is sufficient to show that the cyclic code $\langle h(x)\rangle$ has minimum distance at most $r+1$, where  $h(x)=\frac{x^n-1}{g(x)}$.
Observe that
$$
g(x)=(x^{\frac{n}{r+1}}-1)\prod_{i=\frac{2-d}{2},(r+1)\nmid i}^{\frac{d-2}{2}}(x-\beta^i).
$$
Then \[h(x)\prod_{i=\frac{2-d}{2},(r+1)\nmid i}^{\frac{d-2}{2}}(x-\beta^i)=\frac{x^n-1}{x^{\frac{n}{r+1}}-1}=
\sum_{i=0}^rx^{\frac{(r-i)n}{r+1}}\]
is a codeword of $\langle h(x)\rangle$ and it has Hamming weight $r+1$. This implies that $C^{\perp}$ has minimum distance at most $r+1$.
}
\end{ex}

\subsection{ Proof of part (i) of Theorem \ref{thm:1.1}}.
\begin{proof}
Let $\beta\in\bar{\F}_q$ be an $n$th primitive  root of unity. Put $\alpha=\beta^{\frac{n}{r+1}}$. Then $\Ga\neq 1$ and $\alpha^{q-1}=\left(\beta^n\right)^{\frac{q-1}{r+1}}=1$, i.e., $\Ga$ is an element of $\F_q\setminus\{0,1\}$.
Let $g(x)=(x-1)(x^{\frac{n}{r+1}}-\alpha)$.  

Note that
$$(x^{\frac{n}{r+1}}-\alpha)\left(\sum_{i=0}^r \alpha^{r-i}x^{\frac{in}{r+1}}\right )=x^n-\alpha^{r+1}=x^n-1.$$
This means that $x^{\frac{n}{r+1}}-\alpha$ is a factor of $x^n-1$. Furthermore, since $\gcd(x-1,x^{\frac{n}{r+1}}-\alpha)=1$, we have
$g(x)|(x^n-1)$. Let $C$ be the cyclic code generated  by  $g(x)$.

Let us first show that the locality of $C$ is $r$. By Lemma \ref{lem:3.1}, it is sufficient to show that the Hamming distance of $C^{\perp}$ is at most $r+1$. Put $h(x)=\frac{x^n-1}{g(x)}$. By Lemma \ref{lem:2.1}(i), it is sufficient to show that the Hamming distance of $\langle h(x)\rangle$ is at most $r+1$ since $\langle h(x)\rangle$ and $C^{\perp}$ are equivalent. Consider the codeword
$(x-1)h(x)$ of $\langle h(x)\rangle$. As $(x-1)h(x)=\sum_{i=0}^r \alpha^{r-i}x^{\frac{in}{r+1}}$, the Hamming weight of $(x-1)h(x)$ is $r+1$. Hence, the Hamming distance of $\langle h(x)\rangle$ is upper bounded by $r+1$.

Finally,
as $1$, $\Gb$ are roots of $g(x)$, the minimum distance is at least $3$ by Lemma \ref{lem:2.1}(iii). By the bound \eqref{eq:x1}, the minimum distance of $C$ is upper bounded by $3$.  Thus, the minimum distance of $C$ is exactly $3$. This completes the proof.
\end{proof}
\subsection{ Proof of part (ii) of Theorem \ref{thm:1.1}}.
\begin{proof}
Let $\beta\in\bar{\F}_q$ be an $n$th primitive  root of unity. Put $\alpha=\beta^{\frac{n}{r+1}}$. Then $\Ga\neq 1$ and $\alpha^{q-1}=\left(\beta^n\right)^{\frac{q-1}{r+1}}=1$, i.e., $\Ga$ is an element of $\F_q\setminus\{0,1\}$. Furthermore, $\Ga$ is a $(r+1)$th primitive  root of unity.
Since $\gcd\left(\frac{n}{r+1},r+1\right)$ divides $2$, there exist integers $a,b$ such that $a\times \frac{n}{r+1}+b(r+1)=2$. Put $\Gg=\Ga^a\in\F_q$. Then \[\gamma^{\frac{n}{r+1}}=\Ga^{\frac{an}{r+1}}=\Ga^{2-b(r+1)}=\alpha^2.\]
As $r\ge 3$, we must have $\Gg\in\F_q\setminus\{0,1\}$. Since $\Gg$ is a $(r+1)$th   root of unity, the linear polynomial $(x-\gamma)$ divides $x^n-1$. Furthermore, we have $\gamma^{\frac{n}{r+1}}-\alpha=\alpha^2-\alpha\neq 0$, i.e, $x-\Gg$ is not a factor of $x^{\frac{n}{r+1}}-\alpha$. This implies that $g(x):=(x-1)(x-\gamma)(x^{\frac{n}{r+1}}-\alpha)$ is a divisor of $x^n-1$. Let $C$ be the $q$-ary cyclic code of length $n$ generated by $g(x)$.

 By Lemma \ref{lem:3.1},
to show that the locality of $C$ is $r$, it is sufficient to show that the Hamming distance of $C^{\perp}$ is at most $r+1$. Put $h(x)=\frac{x^n-1}{g(x)}$. Then by Lemma \ref{lem:2.1}(i), it is sufficient to show that the Hamming distance of $\langle h(x)\rangle$ is at most $r+1$ since $\langle h(x)\rangle$ and $C^{\perp}$ are equivalent. Consider the codeword
$(x-1)(x-\Gg)h(x)$ of $\langle h(x)\rangle$. Since $(x-1)(x-\Gg)h(x)=\sum_{i=0}^r \alpha^{r-i}x^{\frac{in}{r+1}}$, the Hamming weight of $(x-1)(x-\Gg)h(x)$ is $r+1$. Hence, the Hamming distance of $\langle h(x)\rangle$ is upper bounded by $r+1$.

Now, we claim that the cyclic code $C$ generated by $g(x)$ has  minimum distance at least $4$.
We prove this claim by contradiction. Suppose that the minimum distance of $C$ were at most $3$. Then  there exists a nonzero polynomial $c(x)=c_0+c_1x^j+c_2x^k\in C$ with $0<j<k\le n-1$ divisible by $g(x)$.
Since  $\beta,\beta^{r+2},\beta^{1+2(r+1)}$ are roots of $x^{\frac{n}{r+1}}-\alpha$ and hence roots of $g(x)$, by Lemma \ref{lem:2.1}(ii) we have
\begin{equation}\label{eq:2}
0=\det\left(
      \begin{array}{ccc}
        1 & \beta^j & \beta^k \\
        1 & \beta^j\beta^{(r+1)j} & \beta^k\beta^{(r+1)k} \\
        1 & \beta^j\beta^{2(r+1)j} & \beta^k\beta^{2(r+1)k} \\
      \end{array}
    \right)=\beta^{j+k}(\beta^{(r+1)j}-1)(\beta^{(r+1)k}-1)(\beta^{(r+1)j}-\beta^{(r+1)k}).
\end{equation}
This forces that $n|j(r+1)$ or $n|(k-j)(r+1)$ or $n|k(r+1)$ since $\beta\in\bar{\F}_q$ is an $n$th primitive  root of unity. We claim that each of these three divisibilities leads to both $n|j(r+1)$ and $n|k(r+1)$.

{\bf Case 1.} $n|j(r+1)$.

Since $1,\beta,\beta^{r+2}$ are also the roots of $c(x)$, by Lemma \ref{lem:2.1}(ii) we have
\begin{equation}\label{eq:3}
\det\left(
      \begin{array}{ccc}
        1 & 1 & 1 \\
        1 & \beta^{j} & \beta^{k} \\
        1 & \beta^{(r+2)j} & \beta^{(r+2)k} \\
      \end{array}
    \right)=(\beta^j-1)(\beta^{k(r+2)}-1)-(\beta^k-1)(\beta^{j(r+2)}-1)=0.
\end{equation}

By using the condition $\beta^{j(r+1)}=1$, Equation \eqref{eq:3} is simplified to
\begin{eqnarray*}0&=&(\beta^j-1)(\beta^{k(r+2)}-1)-(\beta^k-1)(\beta^{j(r+2)}-1)\\
&=&(\beta^j-1)(\beta^{k(r+2)}-1)-(\beta^k-1)(\beta^{j}-1)\\
&=&(\beta^j-1)(\beta^{k(r+2)}-\beta^k)=\Gb^k(\beta^j-1)(\beta^{k(r+1)}-1).\end{eqnarray*}
Observe that $\beta$ is a primitive $n$th root of unity and $0<j<n$. It follows that $\beta^{k(r+1)}=1$, or equivalently $n|k(r+1)$.

{\bf Case 2.} $n|(k-j)(r+1)$.

 By using the condition $\beta^{(k-j)(r+1)}=1$, i.e., $\Gb^{j(r+1)}=\Gb^{k(r+1)}$, Equation \eqref{eq:3} is simplified to
\begin{eqnarray*}0&=&(\beta^j-1)(\beta^{k(r+2)}-1)-(\beta^k-1)(\beta^{j(r+2)}-1)\\
&=&(\beta^j-1)(\beta^{j(r+1)+k}-1)-(\beta^k-1)(\beta^{j(r+2)}-1)\\
&=&(\Gb^k-\Gb^j)(1-\Gb^{j(r+1)})=\Gb^j(\Gb^{k-j}-1)(1-\Gb^{j(r+1)}).\end{eqnarray*}
 Since $0<k\neq j<n$, we must have $n|j(r+1)$. This gives $\Gb^{k(r+1)}=\Gb^{j(r+1)}=1$. Therefore, we have $n|k(r+1)$ as well.

{\bf Case 3.} $n|k(r+1)$.

In this case, we can mimic the proof of Case 1 by swapping $j$ and $k$.

All three cases lead to the conclusion that $n|k(r+1)$ and $n|j(r+1)$. Set $h_1=\frac{j(r+1)}{n}$ and $h_2=\frac{k(r+1)}n$. Then $ h_1\neq h_2\in \{1,2,\dots,r\} $. Moreover, we have the following four identities.
\[\alpha^{h_1}=\Gb^{\frac{n}{r+1}\times h_1}=\Gb^j;\quad \alpha^{h_2}=\Gb^{\frac{n}{r+1}\times h_2}=\Gb^k;\quad \alpha^{2h_1}=\Gg^{\frac{n}{r+1}\times h_1}=\Gg^j;\quad \alpha^{2h_2}=\Gg^{\frac{n}{r+1}\times h_2}=\Gg^k.\]
Finally, we notice that $1,\beta,\gamma$ are the roots of $c(x)$. By Lemma \ref{lem:2.1}(ii), this gives
$$
0=\det\left(
      \begin{array}{ccc}
        1 & 1 & 1 \\
        1 & \beta^j & \beta^k \\
        1 & \gamma^j & \gamma^k \\
      \end{array}
    \right)=\det\left(
      \begin{array}{ccc}
        1 & 1 & 1 \\
        1 & \alpha^{h_1} & \Ga^{h_2} \\
        1 & \alpha^{2h_1} & \Ga^{2h_2} \\
      \end{array}
    \right)=(\alpha^{h_1}-1)(\alpha^{h_2}-1)(\alpha^{h_1}-\alpha^{h_2}).
$$
This is a contradiction  due to fact that $\alpha$ is a $(r+1)$th primitive   root of unity and $ h_1\neq h_2\in \{1,2,\dots,r\} $.

The minimum distance of $C$ is at most $4$ by the bound \eqref{eq:x1}.  The proof is completed.
\end{proof}
\subsection{A locally reparable codes of length $2(q-1)$ and distance $4$}
Finally, we present a $q$-ary optimal \LRC of length $n=2(q-1)$ and minimum distance $4$.

\begin{theorem} If $r+1$ divides $2(q-1)$,
then there exists a   $q$-ary optimal $[n=2(q-1),n-\frac{n}{r+1}-2,4]$ \LRC with locality  $r+1$.
\end{theorem}
\begin{proof} Let $\beta\in\bar{\F}_q$ be a $n$th primitive  root of unity. Then $\Gb^2$ is an element of $\F_q$. Put $\alpha=\beta^{\frac{n}{r+1}}$. Then $\Ga\neq 1$ and $\alpha^{q-1}=\left(\beta^n\right)^{\frac{q-1}{r+1}}=1$, i.e., $\Ga$ is an element of $\F_q\setminus\{0,1\}$.
 Let $g(x)=(x-1)(x-\beta^2)(x^{\frac{n}{r+1}}-\alpha)$. It is straightforward to verify that $g(x)$ is a divisor of $x^n-1$. Let $C$ be the cyclic code with generator polynomial $g(x)$.

By the similar augments in the proof of Theorem \ref{thm:1.1}(i), one can show that dual code of $C$ has minimum distance at most $r+1$.

Since $1,\Gb,\Gb^2$ are roots of $g(x)$, $C$ has distance  at least $4$.

It is easy to check that code $C$ meets the Singleton-type bound \eqref{eq:x1}. The proof is completed.
\end{proof}

\end{document}